\documentclass[lettersize,journal]{IEEEtran}
\usepackage{amsmath,amssymb,amsfonts}
\usepackage{algorithmic}
\usepackage{algorithm}
\usepackage{amsthm}

\usepackage{booktabs}
\usepackage{array}
\usepackage[caption=false,font=normalsize,labelfont=sf,textfont=sf]{subfig}
\usepackage{textcomp}
\usepackage{stfloats}
\usepackage{url}
\usepackage{verbatim}
\usepackage{graphicx}
\usepackage{cite}
\usepackage{soul}
\usepackage{xcolor}
\newcommand{\INPUT}{\item[\textbf{Input:}]}
\newcommand{\INI}{\item[\textbf{Initialization:}]}
\newcommand{\FUN}{\item[\textbf{Function:}]}
\newcommand{\OUTPUT}{\item[\textbf{Output:}]}
\newtheorem{assumption}{Assumption}
\newtheorem{theorem}{Theorem}
\newtheorem{lemma}{Lemma}
\begin{document}


\title{Contextual Combinatorial Beam Management via Online Probing for Multiple Access mmWave Wireless Networks}

\author{
		\IEEEauthorblockN{Zhizhen Li,
            Xuanhao Luo,
            Mingzhe Chen,~\IEEEmembership{Member,~IEEE,}
		Chenhan Xu,~\IEEEmembership{Member,~IEEE,} \\
		Shiwen Mao,~\IEEEmembership{Fellow,~IEEE,} 
            Yuchen Liu,~\IEEEmembership{Member,~IEEE}\\
 }
\thanks{Z. Li, X. Luo, C. Xu, and Y. Liu are with the Department of Computer Science, North Carolina State University, USA (Email: \{zli92, xluo26, chenhanxu, yuchen.liu\}@ncsu.edu). \textit{(Corresponding author: Yuchen Liu.)}}
\thanks{M. Chen is with the Department of Electrical and Computer Engineering and Frost Institute for Data Science and Computing, University of Miami, Coral Gables, FL 33146 USA (Email: \protect\url{ mingzhe.chen@miami.edu)}.} 
\thanks{S. Mao is with the Department of Electrical and Computer Engineering, Auburn University, Auburn,
AL, 36849-5201, USA (Email: smao@ieee.org).}
\vspace{-1.5em}
}



\maketitle

\begin{abstract}
Due to the exponential increase in wireless devices and a diversification of network services, unprecedented challenges, such as managing heterogeneous data traffic and massive access demands, have arisen in next-generation wireless networks. To address these challenges, there is a pressing need for the evolution of multiple access schemes with advanced transceivers. Millimeter-wave (mmWave) communication emerges as a promising solution by offering substantial bandwidth and accommodating massive connectivities. Nevertheless, the inherent signaling directionality and susceptibility to blockages pose significant challenges for deploying multiple transceivers with narrow antenna beams. Consequently, beam management becomes imperative for practical network implementations to identify and track the optimal transceiver beam pairs, ensuring maximum received power and maintaining high-quality access service.
In this context, we propose a Contextual Combinatorial Beam Management (CCBM) framework tailored for mmWave wireless networks. By leveraging advanced online probing techniques and integrating predicted contextual information, such as dynamic link qualities in spatial-temporal domain, CCBM aims to jointly optimize transceiver pairing and beam selection while balancing the network load. This approach not only facilitates multiple access effectively but also enhances bandwidth utilization and reduces computational overheads for real-time applications. Theoretical analysis establishes the asymptotically optimality of the proposed approach, complemented by extensive evaluation results showcasing the superiority of our framework over other state-of-the-art schemes in multiple dimensions.
\end{abstract}


\section{Introduction}

There has been an explosive rise in both the number of mobile users as well as the proliferation of mobile device over the past decade. Such exponential growth is expected to continue due to the rapid evolution of the Internet of Everything (IoE) and innovative mobile applications. 
Emerging technologies such as extended reality (XR)~\cite{ratcliffe2021extended}, holographic video \cite{blinder2019signal}, and the Internet of Vehicles (IoV)~\cite{luo2023clothoid} require high data rates, minimal latency, and support for massive user access, creating unprecedented demands on wireless networks.
These explosive requirements will overwhelm the connection capabilities of the existing fourth generation (4G) and fifth generation (5G) cellular network systems, necessitating the development of next-generation multiple access (NGMA) and advanced transceivers. The primary objective of next-generation (nextG) wireless networks, notably sixth generation (6G), is to revolutionize the network infrastructure, facilitating a vast multitude of users and devices to connect with unparalleled efficiency and flexibility over high-frequency spectrum of radio resources\cite{9693417, jian2020quantitative, gu2024fendi}. This evolution promises to unlock new dimensions of connectivity, seamlessly integrating advanced applications and radio transceivers into our daily lives. 

As a game-changer, millimeter-wave (mmWave) communication stands out in next-generation wireless systems, offering high-bandwidth, low-latency connectivity to address the increasing demands of densely deployed devices and bandwidth-intensive applications, particularly in wireless local-area networks (WLANs). In this domain, beamforming emerges as an advanced transceiver technology due to its capacity to optimize signal strength and reliability in high-frequency, directional radio environments. Unlike traditional omni-directional antennas, which struggle with higher attenuation and interference in mmWave bands, beamforming focuses transmission and reception signals into narrow beams, effectively enhancing signal strength and reducing the interference footprints. Moreover, beamforming enables adaptive beam steering among transceivers, facilitating dynamic communication links and optimizing spectrum utilization, thus ensuring multiple access in mmWave systems.

In mmWave network scenarios, the deployment of multiple transceivers, such as access points (APs), in densely populated WLAN environments is commonplace to meet the demands of bandwidth-intensive applications. Beam management, therefore, becomes imperative for mmWave transceiver implementations to identify and track the optimal transceiver beam pairs, ensuring maximum received power and maintaining high-quality service. However, managing directional beams across numerous transceivers introduces significant challenges in such high path-loss and blockage-prone contexts. First, the task entails formidable overhead, escalating linearly with the number of communication entities. Second, the paired beams are highly susceptible to both static and dynamic obstacles due to the limited range and poor penetration capabilities of mmWave signals. Lastly, the uncertain and time-varying nature of the mmWave channel complicates the adaptive beamforming process, especially in response to changes in line-of-sight (LoS) and non-line-of-sight (NLoS) conditions. This complexity is particularly pronounced in load balancing scenarios where the objective is to maintain consistent service levels across multiple access user equipment.

Several solutions have been developed for beam management and resource allocation for multiple access\cite{10228988,ding2023joint}. However, a key assumption in these prior works is that the channel condition is known from the beginning, which poses a challenge in densely deployed mmWave scenarios characterized by their time-varying nature. Recently, various machine learning-based approaches have emerged to address the uncertainty in beam alignment and selection among transceivers~\cite{zhang2021learning}. For instance, in \cite{polese2021deepbeam}, a deep learning framework is proposed to predict link quality between beams, though its implementation on network devices requires substantial computational resources. 
Alternatively, a multi-armed bandit (MAB) based online learning framework appears more suitable, as it negates the need for offline data collection and strikes a balance between exploration and exploitation in uncertain environments. Such online algorithms can adapt in real-time to changing network conditions, allowing for quick adjustments based on evolving channel characteristics and user demands. In essence, the adaptability is crucial in scenarios where the environment is highly dynamic, such as in mmWave networks where obstacles and interference levels may vary rapidly. 
However, a comprehensive study on a joint transceiver paring and beam management scheme with load balancing in an obstacle-rich mmWave network is still lacking.
Direct application of MAB algorithms like Upper Confidence Bound (UCB) \cite{chen2018contextual,garivier2011upper} may not fully exploit the characteristics of this problem's underlying model, as it overlooks the correlations between nearby beams of transceivers. Given that nearby beams exhibit high spatial correlation, their signaling characteristics such as established link quality are also similar. Therefore, sampling one beam can provide information about its neighboring beams, potentially expediting the convergence to optimal configuration. 

In our prior works\cite{liu2022environment, zzLi}, we introduced a regression-based machine learning framework to predict link quality between mmWave transceivers, accounting for both static and dynamic blockages. This framework has demonstrated an impressive accuracy rate of up to 94\%, requiring minimal environmental data as input. Notably, it seamlessly adapts to different network scenarios by merely modifying input data, eliminating the need for additional model training. By utilizing such link quality predictions as contextual knowledge, the overhead produced in AP probing and beamforming processes can be greatly reduced. Intuitively, APs offering high signal strength at specific locations can be selected for optimal beam pairing. This prior research forms the basis for our current beam management investigation, enabling a context-aware online probing technique tailored to coordination-minimal wireless transceivers.

In this paper, we present a novel contextual combinatorial beam management (CCBM) framework designed to tackle the joint AP and beam selection problem in mmWave wireless networks, ensuring a balanced load distribution among dense transceivers for consistent user services with minimum coordination overhead. In CCBM, each beam is treated as an arm, with the received power serving as reward for selecting specific beams. The objective is to sequentially choose these arms to maximize the cumulative rewards within a given time horizon, particularly allowing user devices to explore \textit{multiple} arms and evaluate their rewards before finalizing the AP-beam selection. This approach significantly minimizes uncertainty by revealing arm rewards before the decision-making process, while minimizing the coordination overhead between transceivers. Additionally, by leveraging link quality predictions of unknown beam directions from our prior works~\cite{liu2022environment, zzLi}, the CCBM framework prioritizes APs based on their predicted link quality at the receiver location. Only beams associated with higher predicted values from these APs are considered during the online probing process. This strategy expedites the assessment of network conditions by avoiding unnecessary searches among irrelevant candidate beams, as both the environmental context and arm context implicitly contribute to the rapid identification of optimal beams. The main contributions of this work\footnote{A part of this paper was presented at the IEEE International Conference on Communications\cite{li2024context}.} are summarized as follows.
\begin{itemize}
    \item We innovatively frame the joint transceiver paring and beam management task as a contextual combinatorial MAB problem, naturally leveraging the correlation between nearby beams and location-aware link qualities as context information to expedite the beam management procedure. 

    \item In our proposed CCBM framework, we incorporate a novel attention-based selection scheme along with an early stopping criterion to prevent excessive exploration during the online probing process. Extensive theoretical analysis establishes an upper bound on the cumulative regret,  i.e., the gap to the results obtained from an oracle search, which demonstrates the asymptotic optimality of our beam management approach.

    \item We develop a reward function within the MAB algorithm that explicitly considers load balancing among candidate APs, guiding receivers in the online probing process to select a globally optimal transmitter and corresponding beam for pairing, thereby optimizing both the communication efficiency and overall network performance.
 
    \item Comprehensive evaluations demonstrate the superiority of our CCBM framework over baseline approaches in various dimensions, including lower regret, increased user throughput, and improved load balancing across densely deployed mmWave APs.

\end{itemize}


\section{Related Works}

Numerous studies have been undertaken to tackle the unprecedented challenges posed by the next-generation multiple access networks, arising from heterogeneous data traffic, massive connectivity demands, and the necessity for ultra-high bandwidth efficiency coupled with ultra-low latency requirements~\cite{9806417}\cite{10278912}. Among these approaches, Non-Orthogonal Multiple Access (NOMA) has emerged as a promising solution~\cite{liu2022developing}, aimed at surmounting the constraints of traditional orthogonal multiple access schemes, and allowing multiple transceivers to share the same resource block through power domain multiplexing, which allocates different power levels to encode data signals for various receivers.  For instance, in \cite{10077113}, a novel cluster-free NOMA framework is proposed for providing scenario-adaptive NOMA communications, utilizing distributed machine learning algorithms for efficient implementation in both single-cell and multi-cell networks. Additionally, in \cite{10278912}, a resource allocation algorithm is developed based on the trust region policy optimization (TRPO) algorithm, addressing the long-term power-constrained sum rate maximization problem in NGMA. Moreover, in \cite{10226210}, a deep neural network (DNN)-based unsupervised learning algorithm is introduced to improve the sum rate performance while maintaining a minimum data-rate requirement in power-domain NOMA setups.

In mmWave networks, beam management plays a crucial role for implementing multiple access techniques like NOMA, as it enables dedicated resource allocation and interference management~\cite{wang2009beam}. Through dynamic adjustments of beam directions and strengths, beam management optimizes transmission quality for multiple access users, thereby maximizing spectral efficiency and system capacity. 
Specifically, a rapid-discovery approach utilizing multi-resolution beam search was introduced in~\cite{wang2009beam}. This technique initially explores a broad beam, progressively refining to narrower beams of a transceiver to determine the optimal beam. Although this method is viable, it requires the adjustment of beam resolution at each stage. {In} \cite{8387219}, a potential map of THz vehicle transmission is developed for autonomous vehicles to address the blockage of short-range and unstable links. In \cite{qi2016coordinated}, the researchers implemented initial access within clustered mmWave small cells by employing the power delay profile. Base stations are arranged into clusters and linked with a backhaul network. They exchange their measurement reports derived from mobile devices and leverage these shared measurements to estimate mobile device locations. As a result, base stations can direct signals towards the estimated locations of transceivers in LoS scenarios. 

Recently, motivated by the remarkable advancements made by deep learning in computer vision and natural language processing, researchers have turned their attention to applying deep learning techniques in the realm of network resource allocation \cite{9044808, ma2023intelligent} and beam management~\cite{liu2020deep,echigo2021deep,qi2020deep}. In \cite{8542687}, a deep learning-based beam management and interference coordination (BM-IC) method is proposed for dense mmWave networks. {In} \cite{9044808}, {a deep Q learning network model is utilized to solve resource allocation problem in 5G architecture. Several DNN models are developed in} \cite{9417452, 9204436} {to predict the best serving beams and facilitate beam training in multiple-input multiple-output (MIMO) systems}. However, a significant drawback of using these deep learning models to optimize beam management is their reliance on large volumes of training data. Additionally, when environmental conditions change, the pre-trained models may become ineffective.  Conversely, online learning approaches offer an adaptive solution by continuously adapting to new incoming or observed data, thus maintaining effectiveness even when the operational environment undergoes significant changes~\cite{qiao2023intelligent}. As an efficient online learning scheme, the multi-armed bandit (MAB) framework has been extensively applied to solve various online optimization problems in the networking area such as channel selection~\cite{6914537,xiao2023resource}, and mobility management\cite{shen2016non}. For example, in the classic MAB model with the UCB algorithm~\cite{auer2002finite}, one arm is played in every round and its corresponding reward is revealed immediately. Then, the objective is to maximize the total expected reward accumulated during $T$ rounds, which can be equivalently formulated as to minimize the regret, i.e., the difference between the reward of the optimal selection and the reward from the algorithm.   

Due to its adaptivity in dynamic scenarios, several MAB algorithms have been developed for beam management in mmWave networks~\cite{va2019online,aykin2020mamba}. For instance, \cite{8485876} adopted a contextual MAB approach to address the beam selection problem in mmWave vehicular systems, leveraging coarse user location information and received data aggregation for environment adaptation. In \cite{8842625}, the beam alignment problem was formulated as a stochastic multi-armed bandit problem, utilizing correlation structure among transceiver beams such that the information from nearby beams is extracted to identify the optimal one. 
Similarly, \cite{hashemi2017efficient} considered beam correlation as arm context with a unimodal beamforming algorithm. However, these existing approaches did not account for practical mmWave network scenarios with numerous obstacles, nor did they jointly consider load-balanced resource allocation in their designs. Notably, \cite{10228988} addressed resource allocation partially with a coarse-level AP probing algorithm, extending the contextual bandit learning framework to handle unknown link rate distributions. However, it did not manage beam pairing for each mmWave transceiver. This work aims to bridge these gaps by jointly considering transceiver paring and beam management while adaptively balancing the network load.

\section{Preliminaries of the Contextual Knowledge Prediction Model}


Given the significance of contextual knowledge in facilitating beam management among mmWave transceivers in a wireless network, our CCBM mechanism is built upon a prediction model that dynamically constructs a comprehensive link quality map of the network environment, considering both spatial and temporal domains.
To ensure precise prediction, our contextual knowledge prediction model comprises two components: long-term and short-term link quality prediction. The long-term prediction aims to forecast a link quality map under static environmental conditions, while the short-term prediction captures changes in link quality caused by environmental dynamics such as moving obstacles.

\subsection{Long-term prediction}
The long-term prediction model is segmented into LoS and NLoS prediction components. This division is crucial due to the significant variation in link quality between these scenarios, particularly pronounced in mmWave bands compared to lower frequencies. Built upon our prior research~\cite{liu2022environment}, the first step involves identifying LoS and NLoS areas within the region. To achieve this, we employ geometric analysis based on~\cite{liu2021maximizing} to delineate these coverage areas. The process basically begins by partitioning the entire network space into a grid of equal-sized 3D cells. Subsequently, each grid cell is scrutinized to determine if the center of the grid falls within a shadowing polygon formed by APs and existing obstacles. In this way, grids devoid of such shadowing polygons are classified as LoS, while grids intersected by shadowing polygons are categorized as NLoS. 

{As widely acknowledged, LoS path component contributes to the majority of link quality at mmWave frequencies. Therefore, the link performance under these scenarios is not highly dependent on the surrounding obstacles, but instead, depends more on the distance between transceivers. As such, LoS link-quality predictions can be performed based on a 3GPP mmWave channel model with parameters chosen for LoS scenarios}~\cite{3GPP16}. 

Forecasting link quality in NLoS areas presents a greater challenge due to the intricate interplay of signal reflection and diffraction effects, which complicate predictions based solely on channel models. Consequently, leveraging deep neural networks (DNNs) becomes a solution to address this complexity and accurately predict link quality between transceivers when an LoS path does not exist.
As introduced in~\cite{liu2022environment}, the proposed DNN model incorporates basic environmental features as input, formulating the prediction task as a regression-based problem. The environmental characteristics encompass various factors, including scenario configuration, obstacle sizes and positions, material reflectivity information, and the locations of APs. By integrating this comprehensive set of static features, the model can effectively capture the nuances of NLoS environments and make precise link quality predictions as verified in our prior work~\cite{liu2022environment}. {Specifically, the input feature is firstly flattened into a vector of size $n_{in}$, which is then fed into a fully connected network with 4 hidden layers, resulting in a signal-to-noise ratio (SNR) value as the output. The model is trained through the backpropagation rule using a mean-squared error loss function.} In essence, this approach harnesses the power of machine learning to navigate the intricacies of multi-path effects between transceivers with paired beams, enabling more robust and reliable forecasts compared to traditional channel modeling techniques alone.

\subsection{Short-term prediction}
Although the long-term link quality map contains information about radio signal strength at various locations relative to the deployed transceivers \cite{luo2024rm}, the prediction model fails to adequately consider the impact of environment dynamics caused by moving obstacles. Because of the mmWave link's susceptibility to blockages, even a slight variation in the location of a moving object can lead to significant changes in the link quality map. This, in turn, affects the beam pairing process among the transceivers. To solve the problem, we further apply a spatial-temporal augmentation model~\cite{zzLi} that is able to predict the \textit{up-to-date} link quality map under dynamic environments. Fig.~\ref{fig:your_label} depicted the overall framework of the short-term prediction model.

\begin{figure}[htbp]
    \centering
    \includegraphics[width=0.9\columnwidth]{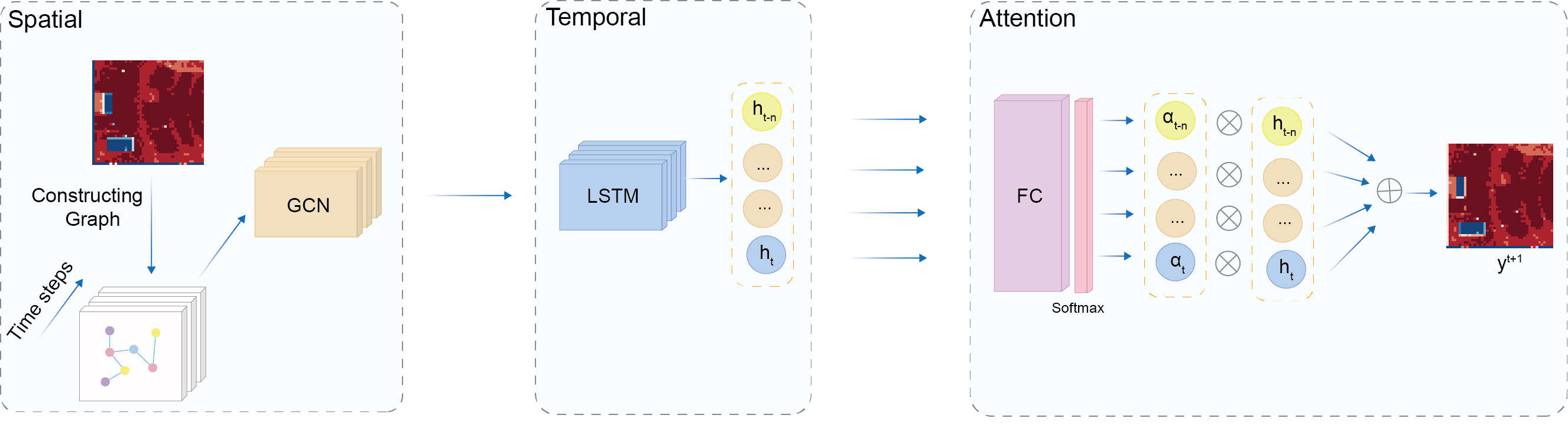}
    \caption{Overview of the short-term prediction model framework.}
    \label{fig:your_label}
\end{figure}

\subsubsection{Spatial model}
Since the presence of moving obstacles can significantly deteriorate the link quality between transceivers situated at random locations, it is crucial to comprehend and model the spatial dependencies that govern the fluctuations within a link quality map. Specifically, each receiver can be regarded as a distinct vertex within a graph and we assume that the neighboring vertices of the receiver share a high degree of correlation in terms of link quality. We then add the edges between these neighbouring vertices to further construct a connected graph, which contains detailed spatial information. To extract meaningful features that accurately represent the environment's impact on link quality, we employ a two-layer Graph Convolutional Networks (GCN) model to extract spatial-domain features, taking into account the graph node and the adjacent links of the node to capture the correlation between link quality and environment details such as deployed objects. A multi-layer GCN can be expressed as:
\begin{equation}
    H^{(l+1)}=\sigma(\Tilde{D}^{-\frac{1}{2}}\Hat{A}\Tilde{D}^{-\frac{1}{2}}H^{(l)}\theta^{(l)}),
\end{equation}
where $\Hat{A}=A+I$ represents the adjacency matrix of the graph enhanced by the identity matrix to include self-connections, ensuring that each node also considers its own state in addition to its neighbors'. $\Tilde{D}$ is the degree matrix with $\Tilde{D}_{ii}=\sum_{j}\Hat{A}_{ij}$ to normalize the graph to account for the varying degrees of connectivity among nodes. $H^{(l)}$ is the output of layer $l$. $\theta^{(l)}$ is the parameter of layer $l$, and $\sigma$ is the activation function.


\subsubsection{Temporal model}
To complement the spatial analysis, we employ a temporal model to understand and predict the temporal fluctuations in link quality. These variances predominantly arise from the combined multipath effects of static and dyanmic blockages, which always evolve over the time horizon. For this purpose, we incorporate a Long Short-Term Memory (LSTM) layer into our short-term prediction framework. As a variant of the recurrent neural network, LSTM is designed to overcome the vanishing gradient problem and make use of the gate mechanism to capture long- and short-term dependencies. 


{In essence, the LSTM model includes three primary gates: the forget gate, input gate and output gate. These gates work together within a memory cell, combining the previous link states and the current environment details to update the hidden link states. The forget gate decides whether the link quality information in the previous memory should be discarded or not. The input gate regulates how much of the new contextual information should be added to the memory cell, while the output gate determines how the memory cell's contents should influence the hidden states. As such, the LSTM layer can effectively predict the link quality of the future time steps 
and capture the dynamic temporal link variations.} 

The integration of spatial and temporal prediction models offers a comprehensive understanding of the complex dynamics governing the mmWave wireless communication between the transceivers. In this setup, the long-term prediction module generates an initial link quality map based on environmental configurations. Subsequently, the short-term engine dynamically updates information at volatile transceiver locations. This synergistic approach ensures a continuous contextual knowledge base for the subsequent beam management process. 

\section{Problem Formulation}

\begin{table}[h]
\centering
\caption{Notations and Definitions.}
\begin{tabular}{|p{0.11\linewidth}|p{0.75\linewidth}|}
\hline
\textbf{Symbol} & \textbf{Description} \\
\hline
$N$ & Number of mmWave APs in a wireless network environment \\
\hline
$C$ & Number of orthogonal beam patterns associated with each mmWave AP \\
\hline
$M$ & Number of clients moving randomly within the space \\
\hline
$X$ & Set of environmental contexts corresponding to user locations \\
\hline
$X'$ & Uniform grid set\\
\hline
$x_m^t$ & Location of user $m$ at time step $t$ \\
\hline
$T$ & Predetermined time horizon \\
\hline
$A$ & Size of the AP candidate set for each user location \\
\hline
$\mathcal{A}_m^t$ & Beam set at time step $t$ for user $m$, consisting of beams from the top-$A$ APs \\
\hline
$B$ & Budget limiting the maximum number of arms that can be probed \\
\hline
$\mathcal{S}_m^t$ & Subset of beams selected to play at time step $t$ for user $m$ \\
\hline
$\mathcal{X}$ & Arm context information set \\
\hline
$r_{a|x}$ & Reward of selecting beam $a$ at user location $x$ \\
\hline
$\mu_{a|x}$ & Expected value of the reward $r_{a|x}$ \\
\hline
$K$ & Penalty weight, which is the maximum number of users that can be connected to a single beam \\
\hline
$k_a$ & Current number of users connected to beam $a$ \\
\hline
$R(\mathcal{S}_m^t, \mathbf{r})$ & Reward of probing subset $\mathcal{S}_m^t$ \\
\hline
$Reg(T)$ & Expected cumulative regret over $T$ rounds \\
\hline
$h_{T}$ & Number of hypercubes\\
\hline
$P^{ue,t}$ & Less explored hypercubes set\\
\hline
$C^t$ & Counter that keeps track of the number
of times the arms within the hypercube has been selected\\
\hline
$n_x$ & The number of times the grid $x$ has been visited in the previous time periods\\
\hline
$K(n_x)$ & Control function that increases with $n_x$ to determine under-explored hypercubes\\
\hline
\end{tabular}
\label{tab:symbols}
\end{table}

In this section, we elaborate on the process of transforming the joint transceiver paring and beam selection problem into a contextual combinatorial MAB problem using link quality information, and then derive an online probing algorithm for effective beam management.
{The choice of MAB over other learning-based algorithms is driven by two main factors: First, the MAB algorithm can quickly adapt to changing environment fast, whereas supervised learning will require retraining under such conditions. Second, compared to deep reinforcement learning, the MAB method requires significantly fewer computational resources, making it more suitable for real-time applications.}
Important notations used in the paper can be found in Table.~\ref{tab:symbols}.

Let $N$ denote the number of mmWave APs in a wireless network environment and $C$ represent the number of orthogonal beam patterns associated with each mmWave transmitter. Additionally, assume that $M$ client receivers are moving randomly within the space. Let $X$ represent the set of environmental contexts corresponding to the user locations. At each time step $t = 1, ..., T$ , where $T$ denotes a predetermined time horizon, the location $x_m^t \in X$ of user $m$ at time $t$ can be observed. Subsequently, the link quality predictions obtained from Sec. III are utilized to rank APs based on the maximum signal strength they can offer at each user location.  We establish an AP candidate set with a size of $A$ for each user location $x_m^t$ by selecting the top-$A$ APs. Considering beam pairing, all the beams from each AP candidate set collectively form a beam set $\mathcal{A}^t_m=\{a_i^j|i\leq C, j \leq A\}$, where $a_i^j$ represents the $i$-th beam of the $j$-th candidate AP. 

Based on the above setup, each beam from $\mathcal{A}^t_m$ can be treated as an \textit{arm} in an MAB problem. At each time step $t$, instead of playing just one arm, a subset of arms $\mathcal{S}^t_m \subset \mathcal{A}^t_m$ will be selected to play. There exists a budget $B$ that limits the maximum number of arms that can be probed, i.e., $|\mathcal{S}^t_m|\leq B$. To choose the optimal subset $\mathcal{S}^t_m$, we incorporate \textit{arm context} information $\mathcal{X}=\{O_a| a\in \mathcal{S}^t_m\}$. Specifically, in our considered scenario, arm context refers to the direction of each beam, where the details about arm selection will be introduced in Sec. V. Here we define the reward of selecting a beam $a$ at the user location $x$ as corresponding to the signal strength of the beam alignment process. We denote this reward by $r_{a|x}$ and its expected value by $\mu_{a|x}=\mathbb{E}[r_{a|x}]$. To model the probing overhead and adhere to the load constraint, the reward of playing a single arm can be further formulated as $\frac{K-k_a}{K}r_{a|x}$, where the penalty weight $K$ is the maximum number of users that can be connected to a single beam, $k_a$ is the current number of users that have connected to beam $a$. Overall, the design of this reward function effectively guides the users to select beams of some APs with lower traffic load while maintaining a relatively high link quality.

As mentioned earlier, in our context, we probe a \textit{subset} of arms $\mathcal{S}^t_m \subset \mathcal{A}^t_m$ to assess the qualities of these arms. We then select the arm that yields the highest reward in $\mathcal{S}^t_m$. Let $\textbf{r}=\{r_{a|x}\}_{a \in \mathcal{S}^t_m}$ denote the collection of rewards of arms in the probing set. The reward of probing $\mathcal{S}^t_m$ can then be formulated as $R(\mathcal{S}^t_m, \textbf{r})$, signifying that the reward is jointly determined by the selection of the subset and the individual reward of each arm in the subset:
\begin{equation}
    R(\mathcal{S}^t_m,\textbf{r})=\max\limits_{a \in \mathcal{S}^t_m} \{\frac{K-k_a}{K}r_{a|x^t_m}\}.
\end{equation}

\noindent {It is worth noting that the $\max$ function in Eq.~(2) is a submodular function}~\cite{10228988}, {which is featured by the diminishing returns property, i.e. given the arm sets $\mathcal{A}$ and $\mathcal{B}$, where $\mathcal{A}\subseteq \mathcal{B}$, for all arms $m \notin \mathcal{B}$, if $R$ is a submodular function, we have:}
\begin{equation}
    R(\mathcal{A}\cup{m},\textbf{r})-R(\mathcal{A},\textbf{r}) \geq R(\mathcal{B}\cup{m},\textbf{r})-R(\mathcal{B},\textbf{r}).
\end{equation}

We prove this with the following theorem: 
\begin{theorem}
    The $\max$ function $R$ in Eq.~(2) is a submodular function
\end{theorem}
\begin{proof}
    Suppose the maximum element in $\mathcal{A}$ and $\mathcal{B}$ are $\alpha$ and $\beta$, respectively. Since $\mathcal{A}\subseteq \mathcal{B}$, we will have $\alpha \leq \beta$. For an element $m \notin \mathcal{B}$, there can be three conditions to be discussed. If $m>\beta$, we have $R(\mathcal{A}\cup{m},\textbf{r})-R(\mathcal{A},\textbf{r})=m-\alpha$, which is greater than $R(\mathcal{B}\cup{m},\textbf{r})-R(\mathcal{B},\textbf{r})=m-\beta$. If $\alpha \leq m \leq \beta$, we will have $R(\mathcal{B}\cup{m},\textbf{r})=R(\mathcal{B},\textbf{r})$ and $R(\mathcal{A}\cup{m},\textbf{r})-R(\mathcal{A},\textbf{r})>0$. The condition still holds. Finally, if $m<\alpha$, we have $R(\mathcal{A}\cup{m},\textbf{r})-R(\mathcal{A},\textbf{r})=R(\mathcal{B}\cup{m},\textbf{r})-R(\mathcal{B},\textbf{r})=0$. Therefore, we can conclude that the $\max$ function will always satisfy the diminishing returns property and is a submodular function.
\end{proof}
This property of the reward function aligns well with our formulated contextual combinatorial MAB problem, which is specifically tailored for submodular functions. 

Given the above property, an online probing method with known expected rewards is outlined in Algorithm 1, {which is a greedy algorithm that always selects the arm to maximize the marginal reward (lines 4-5).}
The primary objective is to maximize the \textit{cumulative} reward expectation over $T$ rounds, i.e., $\sum_{t=1}^T\sum_{m=1}^M\mathbb{E}[R(\mathcal{S}^t_m),\textbf{r}]$. Assuming an optimal algorithm could consistently select the best arm set $\mathcal{S}^{*,t}_m$ at every round $t$ for each user $m$, the performance of our algorithm can be measured by the expected cumulative regret in Eq.~(4), which quantifies the expected cumulative difference between the maximum reward achieved by the optimal algorithm and the reward obtained by Algorithm 1.
\begin{equation}
    Reg(T)=\sum_{t=1}^T\sum_{m=1}^M\mathbb{E}[R(\mathcal{S}^{t,*}_m,\textbf{r})-R(\mathcal{S}^t_m,\textbf{r})]. 
\end{equation}
{It is worth noting that there is no gap between the transition of these two evaluation metrics. As the algorithm that provides the best cumulative reward will definitely yield a smaller difference to the optimal reward. Hence, the objective of maximizing the expected cumulative reward is equivalent to minimizing the expected regret in our problem.}

It has been proven that maximizing a submodular set function with known reward expectation is NP-hard~\cite{goel2006asking}. However, a greedy probing algorithm has been proposed in~\cite{nemhauser1978analysis} that guarantees achieving no less than $(1-1/e)$ of the optimal solution. Therefore, as described in Algorithm~\ref{alg:greedy}, beams can be sequentially selected based on their marginal reward to ensure an asymptotic optimality.
Given that no polynomial time algorithm can achieve a better approximation for the submodular function maximization problem, our objective is to find an algorithm that achieves sublinear $(1-\frac{1}{e})$-approximation regret, as formulated in Eq. (5):

\begin{equation}
    Reg(T)=(1-\frac{1}{e})*\sum_{t=1}^T\sum_{m=1}^M\mathbb{E}[R(\mathcal{S}^{t,*}_m),\textbf{r}]-\sum_{t=1}^T\sum_{m=1}^M\mathbb{E}[R(\mathcal{S}^t_m,\textbf{r})]. 
\end{equation}

\begin{algorithm}
    \caption{Online Probing Algorithm}
    \label{alg:greedy}
    \begin{algorithmic}[1]
        \INPUT arm set $\mathcal{A}$, reward function $R$, budget $B$
        \OUTPUT super arm $\mathcal{S}$
        \STATE $\mathcal{S} \leftarrow \varnothing$;
        \STATE $i=0$;
        \WHILE{$i \leq B$}
            \STATE $m=\underset{m \in \mathcal{A} \backslash
\mathcal{S}}{\mathrm{argmax}}\, R(\mathcal{S} \cup \{m\},\textbf{r})-R(\mathcal{S},\textbf{r}) $;
            \STATE $\mathcal{S}=\mathcal{S}\cup\{m\}$;
            \STATE $i=i+1$;
        \ENDWHILE

    \end{algorithmic}
\end{algorithm}

\section{Contextual Combinatorial Beam Management}
Built upon the problem formulation as discussed in Sec.~IV, this section introduces a novel contextual combinatorial MAB approach for beam management in mmWave wireless networks. In practical scenarios, it will be infeasible to obtain prior knowledge of the expected rewards for arms. Consequently, direct application of Algorithm~\ref{alg:greedy} is not viable. Instead, we aim to learn the expected values of arms using a contextual combinatorial MAB framework as illustrated in Sec.~IV. Such a framework was originally designed for general bandit problems with submodular reward function~\cite{chen2018contextual}. However, our approach differs by incorporating both the arm context (beam correlation) and the environment context (location-aware link qualities). Additionally, we subtly integrate a beam selection scheme to enhance rewards during the exploration period. 

\begin{figure}[htbp]
\centerline{\includegraphics[scale=0.1]{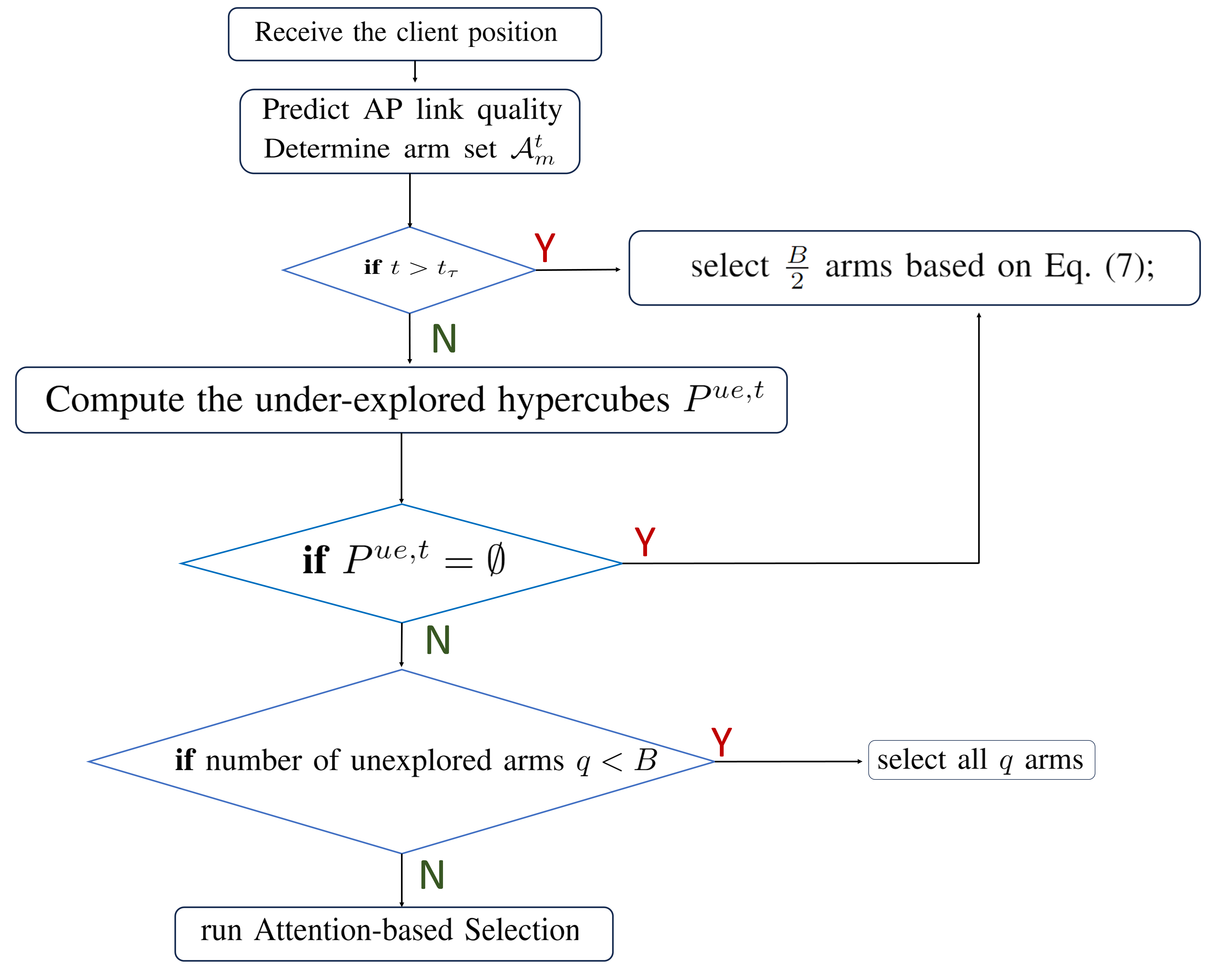}}
\caption{{Workflow of CCBM procedure.}}
\label{diagram}
\end{figure}

Algorithm 2 summarizes our CCBM framework using online probing to achieve the expected rewards for beam arms. {The overall workflow is shown in} Fig.~\ref{diagram}. Initially, the environment context space $X$ is divided into a uniform grid set $X'=\{x_1,x_2,...,x_n\}$, with $n$ representing the total number of grids into which the space is partitioned. At each time step $t$, when a user location $x^t_m$ is observed, it is mapped to the corresponding grid $x$ in $X'$ to which it belongs (Lines 4-5). In addition to managing the environment context space, we partition the arm context space $\mathcal{X}=[0,2\pi]$, where it is first normalized to $[0,1]$, and then divided into $h_T$ \textit{hypercubes}, each with a size of $\frac{1}{h_T}$. 

Specifically, at each time step $t$, for each user $m$, the algorithm observes the environment context $x^t_m$ (e.g., the user's location) and then maps it to the corresponding grid $x \in X'$. {Later, the link quality prediction model in Sec. III is utilized to predict link quality of each AP given the position information. The candidate APs are chosen based on their ranking of the predicted link quality values. Then, the beams from all the candidate APs form the arm set $\mathcal{A}^t_m$ (Lines 6-7).} For each arm $a$ in $\mathcal{A}^t_m$ with arm context $O_a$, the algorithm determines a hypercube $p_a \in \mathcal{X}$ such that $O_a \in p_a$ holds. The collection of hypercubes at time slot $t$ is denoted as $\textbf{p}^t={p_a}_{a \in \mathcal{A}^t_m}$ (Lines 12-13). Subsequently, the algorithm identifies the hypercubes $p_a \in \textbf{p}^t $ that are explored less frequently based on the following criteria:
\begin{equation}
    P^{ue,t}=\{p_a \in \textbf{p}^t| x \in X', a \in \mathcal{A}^t,  C^t(p_a|x)<K(n_x)\},
\label{equ5}
\end{equation}
where $C^t(p_a|x)$ is a counter that keeps track of the number of times the arms within the hypercube $p_a$ are selected when the receiver location is mapped to $x$ during time periods $1,2,..,t-1$. $K(n_x)$ is a deterministic, monotonically increasing control function, and $n_x$ represents the number of times grid $x$ has been visited in the previous time periods.

Next, the algorithm determines whether to explore or exploit based on the number of arms located in under-explored hypercubes. If the set of under-explored arms is non-empty, the algorithm enters an exploration phase. Let $q$ be the size of under-explored arm set. If the under-explored arm set contains at least $B$ arms, i.e., $q \geq B$, we employ an arm selection scheme called \textit{attention-based selection} (Lines 31-37) instead of randomly selecting arms as in prior MAB works.

\begin{algorithm}[t]
    \caption{Contextual Combinatorial Beam Management}
    \label{alg:max_of_two}
    \begin{algorithmic}[1]
        \INPUT user number $M$, arm set $\mathcal{A}^t_m$, reward function $R$, budget $B$, time horizon $T$, control function $K(n_x)$, arm context space $\mathcal{X}$, grid set $X$
        \INI $\forall x \in X, n_x=0$\\$\forall p_a \in \mathcal{X}, C^t(p_a|x)=0, \hat{r}(p_a|x)=0$
        \FOR{$t=1,2,\dots,T$}
        \FOR{$m=1,2,\dots,M$}
        \STATE $\mathcal{S}^t_m=\emptyset$;
        \STATE Receive the client position information $x^t_m$;
        \STATE Map the position context to grid $x\leftarrow x^t_m$;
        \STATE {Predict link quality based on the positions of APs};
        \STATE {Determine arm set $\mathcal{A}^t_m$ based on the predicted AP link quality};
        \STATE $n_x=n_x+1$;
        \IF{$t>t_\tau$}
        \STATE $\mathcal{S}^t_m \leftarrow$ select $\frac{B}{2}$ arms based on Eq.~(\ref{equ6});
        \ELSE
        \STATE Find $\textbf{p}^t$, such that $\forall a \in \mathcal{A}^t_m$, $O_a \in p_a, p_a \in \mathcal{X}$;
        \STATE Compute the under-explored hypercubes $P^{ue,t}$ using Eq.~(\ref{equ5});
        \IF{$P^{ue,t}= \emptyset$}
        \STATE $\mathcal{S}^t_m \leftarrow$ select $B$ arms based on Eq.~(\ref{equ6});
        \ELSE
        \IF{number of unexplored arms $q <B$}
        \STATE $\mathcal{S}^t_m \leftarrow$ select all $q$ arms and the other $B-q$ arms based on Eq.~(\ref{equ6});
        \ELSE
        \STATE run Attention-based Selection(void);
        \ENDIF
        \ENDIF
        \ENDIF
        \FOR{each arm $a \in \mathcal{S}^t_m$}
        \STATE observe the quality $r_{a|x}$ of $a$; 
        \STATE update $\hat{r}(p_a|x)=\frac{\hat{r}(p_a|x)C^t(p_a|x)+r_{a|x}}{C^t(p_a|x)+1}$;
        \STATE update counter $C^t(p_a|x)=C^t(p_a|x)+1$;
        \ENDFOR
        \ENDFOR
        \ENDFOR

        \FUN Attention-based Selection(void):
        \IF{$|\mathcal{Z}|\geq B$}
        \STATE $\mathcal{S}^t_m \leftarrow$ randomly select arms from $\mathcal{Z}$;
        \ELSIF{$0<|\mathcal{Z}|<B$}
        \STATE  $\mathcal{S}^t_m \leftarrow$ select all arms in $\mathcal{Z}$ and other arms randomly;
        \ELSE
        \STATE $\mathcal{S}^t_m \leftarrow \{a^{t-1}_m\} \cup$$\{ B-1$ arms randomly selected\};
        \ENDIF

    \end{algorithmic}
\end{algorithm}

\vspace{+0.1cm}
\textbf{Attention-based Selection}: Let $\mathcal{Z}$ be the set of arms in the under-explored arm set with $C^t(p_a|x)=0$, indicating that at grid $x$, the hypercube to which arm $a$ belongs to has never been chosen until time period $t-1$. If $|\mathcal{Z}| \geq B$, then randomly select $B$ arms from $\mathcal{Z}$. If $0< |\mathcal{Z}| <B$, select all the arms in set $\mathcal{Z}$ and randomly select other under-explored arms. The rationale behind this step is intuitive: If we only randomly select arms without attentions, there is a possibility that certain hypercubes providing good rewards may not be identified in the initial rounds, leading to suboptimal exploration. 

To be specific, when $|\mathcal{Z}|=0$, indicating that the under-explored hypercubes have been chosen at least once, our algorithm first identifies the arm $a_m^{t-1}$ chosen for the user $m$ in the last time step $t-1$. Since we are considering a continuous movement (action space), for a single user, the location at time step $t$ should be close to the location at time step $t-1$. Therefore, we can still assume $a_m^{t-1}$ is a good candidate arm that can provide satisfactory rewards at round $t$, and it will be included in the probing set $\mathcal{S}^t_m$, while the other arms are chosen randomly. In this way, we strategically incorporate attention-based exploitation into the exploration phase.

In some cases, if the under-explored arm set contains fewer than $B$ elements, i.e., $q\leq B$, then the algorithm selects all $q$ arms (Lines 17-18). The remaining arms are selected sequentially by exploiting the estimated rewards as follows:
\begin{equation}
a=\underset{a \in \mathcal{A}^t_m \backslash
\mathcal{S}^t_m}{\mathrm{argmax}}\, R(\mathcal{S}^t_m \cup \{m\},\hat{\textbf{r}})-R(\mathcal{S}^t_m,\hat{\textbf{r}}), 
\label{equ6}
\end{equation}
where $\hat{r}$ is used to denote the sampled reward of each arm $a$. If there are no under-explored arms, all $B$ arms will be selected based on Eq.~(\ref{equ6}). 

Since we consider an mmWave wireless network scenario with fixed APs, it is intuitive that after a certain number of rounds of exploration, we can have a relatively comprehensive knowledge of the network condition. Thus, it will be more rewarding to perform exploitation after a certain time step. To this end, we incorporate a \textit{early stopping criterion} to guide the algorithm into an exploitation phase (Lines 9-10).

\vspace{+0.1cm}
\textbf{Early Stopping Criterion}: We assume that after a time threshold $t_\tau$ 
the algorithm enters a pure exploitation period. Since arms yield different rewards in terms of different grids $x$, if all the grids have been visited by users several times, the network condition can be well revealed. Thus, we set $t_\tau$ equal to the number of grids $n$ across the space. It is also worth noting that we reduce the size of the probing set to $\frac{B}{2}$ during the pure exploitation period. Numerical results in Sec.~VII will show that this added criterion greatly reduces the beam search overhead while maintaining a competitive reward.

\section{Theoretical Analysis}
In this section, we provide a rigorous theoretical analysis of the regret bound using our CCBM approach to achieve the optimal beam paring among mmWave transceivers. 
The upper bound is derived under the principle that arms belonging to similar context space should have similar expected reward values.

\begin{assumption}
(Bounded Reward) The reward of each arm is bounded by $0<r<r^{max}$.
\end{assumption}
\begin{assumption}
(Lipschitz-continuous) There exists $C > 0$ such that for any arm $a,a'$ with arm context $O_a,O_{a'}\in \mathcal{X}$, we have $|r_a-r_{a'}|\leq C||O_a-O_{a'} ||_1$.
\end{assumption}

It is worth noting that these two assumptions are mild assumptions. As the reward of selecting each arm refers to the signal strength at each arm, it is easy to follow that the reward is bounded. For Assumption 2, since the arm context represents the orientation of beams and the reward is bounded, we can always find $C$ such that the Lipschitz-continuous assumption holds.

We set the $h_T=\lceil T^{\frac{1}{4}} \rceil$ for the arm context partition and $K(n_x)=n_x^{\frac{1}{2}}\log(n_x)$ as the control function to identify the under-explored arm hypercubes. Then, the regret can be bounded as follow:
\begin{theorem}
    Let $h_T=\lceil T^{\frac{1}{4}} \rceil$ and $K(n_x)=n_x^{\frac{1}{2}}\log(n_x)$, if Assumptions 1 and 2 hold true, the regret $R(T)$ is bounded by:
    
   \begin{align}
R(T) &\leq (1 - \frac{1}{e})B r^{\text{max}} \left(2 M(M^{\frac{1}{2}}T^{\frac{3}{4}} \log(MT) + T^{\frac{1}{4}})\right) \nonumber \\
&\quad + \left(1 - \frac{1}{e}\right) \cdot MB^2r^{\text{max}} \left(\begin{array}{c} M^{\text{max}} \\ B \end{array}\right) \frac{\pi ^2}{3} \nonumber \\
&\quad + \left(3BL + \frac{8}{3}B(r^{\text{max}}+L)\right)T^{\frac{3}{4}}.
\end{align}

\end{theorem}
\begin{proof}
{The regret here is the same as we introduced before, which is the expected cumulative difference between the maximum reward achieved by the optimal algorithm and the reward obtained by proposed algorithm.} The regret \( R(T) \) can be divided into the following summands:
\[ \mathbb{E}[R(T)] = \mathbb{E}[R_{explore}(T)] + \mathbb{E}[R_{exploit}(T)],\]
where the term \( \mathbb{E}[R_{explore}(T)] \) is the regret due to the exploration process, and the term \( \mathbb{E}[R_{exploit}(T)] \) corresponds to the regret in the exploitation phase. We first derive a bound on \( \mathbb{E}[R_{explore}(T)] \). According to Algorithm 2, the set of under-explored hypercubes \( P^{\text{ue},t}_T \) is non-empty during the exploration phase, which implies that there exists at least one hypercube \( p \) with  \( C^t(p|x) \leq K(n_x) = n_x^{\frac{1}{2}} \log(n_x) \). Because we only explore in the first
$t_\tau$ rounds, \(n_x < Mt_\tau <MT \) holds. Certainly, there can be a maximum of \( \lceil(MT)^{\frac{1}{2}} \log(MT)\rceil \) exploration phases in which \(p\) is under-explored. Given \( h_T \) hypercubes in the partition and a total of \(M\) users, the upper limit for the exploration phases is \( h_TM \lceil (MT)^{\frac{1}{2}} \log(MT) \rceil \). Owing to the submodularity of reward function and its bounded nature, the maximum regret for an incorrect selection in one exploration phase is constrained by \( (1 - 1/e)B r^{\text{max}} \). Therefore, we have
\begin{align*}
\mathbb{E}[R_{explore}(T)] \leq (1 - \frac{1}{e})B r^{\text{max}} h_TM \lceil(MT)^{\frac{1}{2}} \log(TM)\rceil\\
= (1 - \frac{1}{e}) B r^{\text{max}}M \lceil T^{\frac{1}{4}} \rceil \lceil(MT)^{\frac{1}{2}} \log(TM)\rceil.
\end{align*}
Given \( \lceil T^{\frac{1}{4}}\rceil \leq 2T^{\frac{1}{4}}\), we can further bound the maximum regret as:
\begin{align*}
\mathbb{E}[R_{explore}(T)] \leq (1 - \frac{1}{e})B r^{\text{max}}2 M(M^{\frac{1}{2}}T^{\frac{3}{4}} \log(MT) + T^{\frac{1}{4}}).
\end{align*}

Prior to establishing the limit on the expected value of $\mathbb{E}[R_{exploit}(T)]$, we first introduce some auxiliary functions. For each hypercube \( p  \), we define $ \bar{\mu}(p) = \sup_{O \in p} \mu(O) \quad \text{and} \quad \underline{\mu}(p) = \inf_{O \in p} \mu(O) $ to represent the best and worst expected quality over all contexts \( O \in p \). Furthermore, the context at center of a hypercube \( p \) is defined as \( \hat{O}_p \) and its expected quality is \( \hat{\mu}(p) = \mu(\hat{O}_p) \). Let
$\bar{\mathbf{\mu}}_p^t = [\bar{\mu}(p_1^t), \ldots, \bar{\mu}(p_{h_T}^t)], \quad \underline{\mu}_p^t = [\underline{\mu}(p_1^t), \ldots, \underline{\mu}(p_{h_T}^t)], 
 \tilde{\mu}_p^t = [\tilde{\mu}(p_1^t), \ldots, \tilde{\mu}(p_{h_T}^t)], $ and define \( \tilde{S}^{*,t}(p^t) \) as:
 \[ \tilde{S}^{*,t}(p^t) = \underset{S \subseteq \mathcal{A}^t_m, |S| \leq B}{\arg\max} \ R(S,\tilde{\mu}_p^t). \]
Let $\tilde{S}^{*,t}(p^t)$ be the optimal set and \( \tilde{S}^t(p^t) \) be the set that is chosen by Algorithm 1. We will have $R(\tilde{S}^t(p^t),\tilde{\mu}_{p}^t ) \geq (1 - 1/e) \cdot R(S^{*,t}(p^t),\tilde{\mu}_{p}^t ).$ The arm set \( \tilde{S}^t(p^t) \) assists in identifying the subsets of arms which are sub-optimal. Let
\[
\mathcal{L}^t(p^t) = \{ G \subseteq \mathcal{A}^t_m, |G| = B : R(\tilde{S}^t(p^t),\underline{\mu}^t_{p} ) - R(G,\bar{\mu}_{p}^t) \geq An_x^{\theta} \}
\]
be the collection of suboptimal subsets of arms for hypercubes \( p^t \), where \( A > 0 \) and \( \theta < 0 \). A subset \( G \) of arms in \( \mathcal{L}^t(p^t) \) is considered suboptimal for \( p^t \), since the sum of the worst expected reward in \( \tilde{S}^t(p^t) \) is at least an amount \( An_x^\theta \) higher than the sum of the best expected reward for subset \( G \). Subsets in \( S_B^t \backslash \mathcal{L}^t(p^t) \) is regarded as near-optimal for \( p^t \), where \( S_B^t \) denotes the set of all \( B \)-element subsets of arm set \( M^t \). Then, $\mathbb E[R_{exploit}(T)]$ can be divided into 
\[ \mathbb E[R_{exploit}(T)]= \mathbb{E}[R_s(T)] + \mathbb{E}[R_n(T)] \]

\noindent where  \( \mathbb{E}[R_s(T)] \) is the regret due to suboptimal choices, i.e., the subsets of arms from \( \mathcal{L}^t(p^t) \) are selected; \( \mathbb{E}[R_n(T)] \) is the regret due to near-optimal choices, i.e., the subsets of arms from \( S_B^t \backslash \mathcal{L}^t(p^t) \) are selected. We will prove the bound of each term in the following. We first derive the bound for $\mathbb{E}[R_s(T)]$.

For time slot \(1 \leq t \leq T\), denote \(W^t = \{\mathcal{P}^{ue,t} = \emptyset\}\) as the scenario where slot \(t\) is an exploitation phase. According to the definition of $\mathcal{P}^{ue,t}$, under this condition, it holds that \(C^t(p^t_m) > K(n_x) = n_x^{\frac{1}{2}} \log(n_x), \forall p \in p^t\). Define \(V^t_G\) as the occurrence that subset \(G \in \mathcal{L}^t(p^t)\) is selected at time slot \(t\). Then, it holds that

\vspace{-0.2cm}
\begin{align*}
R_s(T) = \sum_{t=1}^{T}\sum_{m=1}^{M} \sum_{G \in \mathcal{L}^t(p^t)} I_{\{V^t_G,W^t\}} \times \\
\left( (1 - \frac{1}{e})R( S^{*,t}(x^t),r^t) - R( G,r^t) \right),
\end{align*}

\noindent where for each slot, we assess the performance decrement resulting from opting for a non-ideal set of arms \(G \in \mathcal{L}^t(p^t)\). Since the maximal potential performance drop by choosing \(G\) is bounded by \((1 - \frac{1}{e})Br^{\text{max}}\), we have

\[
R_s(T) \leq (1 - \frac{1}{e})Br^{\text{max}} \sum_{t=1}^{T} \sum_{G \in \mathcal{L}^t(p^t)} I_{\{V^t_G,W^t\}}. 
\]
\noindent  The expected regret is constrained by
\[
\mathbb{E}[R_s(T)] \leq (1 - \frac{1}{e})B r^{\text{max}} \sum_{t=1}^{T} \sum_{G \in \mathcal{L}^t(p^t)} \mathbb{E} [I_{\{V^t_G, W^t\}}] 
\]
\[
= (1 - \frac{1}{e})B r^{\text{max}} \sum_{t=1}^{T} \sum_{G \in \mathcal{L}^t(p^t)} \text{Prob} \{V^t_G, W^t\}. 
\]
In the situation where event \(V^t_G\) takes place, according to the algorithm, the reward for choosing arms in \(G\) surpasses or equals that of selecting arms in \(\tilde{S}^t(p^t)\), signifying \(R( G,\hat{r}^t_p) \geq R( \tilde{S}^t(p^t),\hat{r}^t_p)\). Consequently, we have:
\[
\text{Prob} \{V^t_G, W^t\} \leq \text{Prob} \left\{ R( G,\hat{r}^t_p) \geq R(\tilde{S}^t(p^t),\hat{r}^t_p) \right\} 
\]
The condition on the right-hand side indicates the occurrence of at least one of the subsequent events for any \(H(n_x) > 0\):
\[
E_1 = \{R(G,\hat{r}^t_p) \geq R(G,\bar\mu^t_p,) + H(n_x), W^t \}
\]
\[
E_2 = \{R(\hat{S}^t(p^t),\hat{r}^t_p) \leq R(\hat{S}^t(p^t),\underline\mu^t_p) - H(n_x), W^t \}
\]

\noindent Hence, for the original event, it follows that:
\[
\left\{ R( G,\hat{r}^t_p) \geq R( \hat{S}^t(p^t),\hat{r}^t_p) \right\} \subseteq E_1 \cup E_2 
\]

We assess the probabilities of the two events \( E_1 \) and \( E_2 \) distinctly. We will first start by bounding \( E_1 \). Recall that the best expected quality for arms within the set \( p \) is defined by \( \bar{\mu}(p) = \sup_{O \in p} {\mu}(O) \). Hence, the expected quality of arm \( m \) in \( G \) is constrained by

\begin{align*}
E[\hat{r}(p^t_m)] &= \mathbb{E}\left[
  \frac{1}{|\mathcal{E}^t(p^t_m)|} \sum_{(\tau, k) : O_k^T \in {p^T_m}, n \in \mathcal{S}^T} r\left(O_k^\tau\right)
\right] \\
&= \frac{1}{|\mathcal{E}^t(p^t_m)|} \underbrace{\sum_{(\tau, k) : O_k^T \in {p^T_m}, n \in \mathcal{S}^T}}_{|\mathcal{E}^t(p^t_m)|summands} \underbrace{\mu\left(O_k^\tau\right)}_{\leq \bar{\mu}(p^t_m)}\\
&\leq \bar{\mu}(p^t_m).
\end{align*}

\noindent This deduction suggests that
\begin{align*}
\text{Prob}\{E_1\} &= \text{Prob}\{ R(G,\hat{r}^t_p) \geq R(G,\bar{\mu}^t_p) + H(n_x), W^t\} \\
&\leq \text{Prob}\{ \hat{r}(p_m^t) \geq \bar{\mu}(p_m^t) + \frac{H(n_x)}{B}, \exists m \in G, W^t\} \\
&\leq \text{Prob}\{ \hat{r}(p_m^t) \geq \mathbb{E}[\hat{r}(p_m^t)] + \frac{H(n_x)}{B}, \exists m \in G, W^t\} \\
&= \sum_{m \in G} \text{Prob}\{ \hat{r}(p_m^t) \geq \mathbb{E}[\hat{r}(p_m^t)] + \frac{H(n_x)}{B}, W^t\}.
\end{align*}

The rationale behind the first inequality lies in the proposition that $\{R(G,\tilde{r}^t_{p}) \geq R(G,\bar{\mu}^t_{p}) + H(n_x)\} \subseteq \{ \hat{r}(p_m^t) \geq \bar{\mu}(p_m^t) + \frac{H(n_x)}{B}, \exists m \in G \}$. This assertion can be confirmed through a direct application of contradiction and the principle of submodularity. Now, we can apply the Chernoff-Hoeffding bound using Lemma 1:
\begin{lemma}
    Let $X_1, X_2, \ldots, X_n$ be independent random variables bounded by the interval $[0, 1]$, i.e., $0 \leq X_i \leq 1$ for all $i = 1, 2, \ldots, n$. Let $\bar{X} = \frac{1}{n}\sum_{i=1}^{n}X_i$ be the sample mean and let $\mu = \mathbb{E}[\bar{X}]$ be the expected value of the sample mean. Then, for any $\epsilon > 0$, the following bounds hold:

\begin{enumerate}
    \item Upper Tail: $\Pr(\bar{X} - \mu \geq \epsilon) \leq \exp\left(-2n\epsilon^2\right)$.
    \item Lower Tail: $\Pr(\bar{X} - \mu \leq -\epsilon) \leq \exp\left(-2n\epsilon^2\right)$.
\end{enumerate}
\end{lemma}
Note that according to assumption 1, the estimated quality of each arm is bounded by \( r^{\text{max}} \) and recognizing that the occurrence of event \( W^t \) suggests that a minimum of \( n_x^{\frac{1}{2}}\log(n_x) \) samples were drawn. We have
\[
\text{Prob}\{E_1\} = \sum_{m \in G} \text{Prob}\{ \hat{r}(p_m^t) \geq \mathbb{E}[\hat{r}(p_m^t)] + \frac{H(n_x)}{B}, W^t\}
\]

\[
~~~~~~~~~~~~~= \sum_{m \in G} \text{Prob}\{ \hat{r}(p_m^t) - \mathbb{E}[\hat{r}(p_m^t)] \geq \frac{H(n_x)}{B}, W^t\}
\]

\[
~~~~~~~~~~~~~= \sum_{m \in G} \text{Prob}\{ \frac{\hat{r}(p_m^t)}{r^{max}} - \frac{\mathbb{E}[\hat{r}(p_m^t)]}{r^{max}} \geq \frac{H(n_x)}{Br^{max}}, W^t\}
\]

\[
~~~~\leq \sum_{m \in G} \exp \left( \frac{-2H(n_x)^2n_x^{\frac{1}{2}}\log(n_x)}{B^2(r^{\text{max}})^2} \right).
\]

{In the third equality, both sides are divided by $r^{max}$ to ensure that the range of $\frac{\hat{r}(p_m^t)}{r^{max}}$ and $\frac{\mathbb{E}[\hat{r}(p_m^t)]}{r^{max}}$ are both within [0, 1], satisfying the conditions of Lemma 1. Therefore, Lemma 1 can be directly applied, where $\epsilon$ is set to $\frac{H(n_x)}{Br^{max}}$.}

Similarly, for event \(E_2\), the estimation follows
\[
\text{Prob}\{E_2\} = \text{Prob}\{R(\tilde{S}^t(p^t),\hat{r}^t_{p}) \geq R(\tilde{S}^t(p^t),\underline{\mu}^t_{p}) - H(n_x), W^t\}.
\]

\[
\leq \sum_{m \in S^t(p^t)} \exp \left( \frac{-2H(n_x)^2n_x^{\frac{1}{2}}\log(n_x)}{B^2(r^{\text{max}})^2} \right)
\]

By far, the analysis is conducted with an arbitrary \( H(n_x) > 0 \). We will select \( H(n_x) = Br^{\text{max}}n_x^{-\frac{1}{4}}\sqrt{1-\frac{\log(B)}{2\log(n_x)}} \). Then, we have
\begin{align*}
\text{Prob}\{E_1\} &\leq B \exp \left( -\frac{2H(n_x)^2n_x^{\frac{1}{2}} \log(n_x)}{B^2(r^{\text{max}})^2} \right) \\
&\leq B \exp(-2 \log(n_x) -\log(B) \\
&\leq n_x^{-2} \leq t^{-2}
\end{align*}
and similarly
\begin{align*}
\text{Prob}\{E_2\} \leq n_x^{-2} \leq t^{-2}
\end{align*}

To sum up,
\begin{align*}
&\text{Prob} \{ V^t_G, W^t \} \leq \text{Prob} \{E_1 \cup E_2 \}\\
&\leq \text{Prob} \{E_1\} + \text{Prob} \{E_2\} \leq 2t^{-2}.   
\end{align*}

Given this, we have:
\[
E[R_s(T)] \leq \left(1 - \frac{1}{e} \right)Br^{\text{max}} \times \sum_{t=1}^{T} \sum_{m=1}^{M} \sum_{G \in \mathcal{L}(p^t)} \text{Prob} \{ V^t_G, W^t \}
\]
\[
~~\leq \left(1 - \frac{1}{e} \right)Br^{\text{max}} \left( \begin{array}{c} |\mathcal{A}^t_m| \\ B \end{array} \right)M \sum_{t=1}^{T} 2t^{-2}
\]
\[
~~~~~~\leq \left(1 - \frac{1}{e} \right)B(r^{\text{max}}) \left( \begin{array}{c} |\mathcal{A}^t_m| \\ B \end{array} \right) \cdot 2M \sum_{t=1}^{\infty} t^{-2}
\]
\[
\leq \left(1 - \frac{1}{e} \right)B(r^{\text{max}})M \left( \begin{array}{c} |\mathcal{A}^t_m| \\ B \end{array} \right) \frac{\pi^2}{3}, 
\]
where $\left( \begin{array}{c} |\mathcal{A}^t_m| \\ B \end{array} \right)$ denotes the maximum number of unique subsets of size \textit{B} can be chosen from $|\mathcal{A}^t_m|$. For the last step, we have:

\[  \sum_{t=1}^{\infty} t^{-2} = \frac{\pi^2}{6}  \].
This can be proven by using Taylor series expansion of $\sin (t)$:
\[ \sin(t) = t - \frac{t^3}{3!} + \frac{t^5}{5!} - \frac{t^7}{7!} + \cdots.\]
Since $\sin(t) \text{ has zeros at } t = 0, \pm \pi, \pm 2\pi, \pm 3\pi, \ldots$, it can be also represented as:
\[ \sin(t) = t \prod_{n=1}^{\infty} \left(1 - \frac{t^2}{n^2 \pi^2}\right)
\]
By expanding the infinite product and comparing the coefficient of $t^3$ from both the product and the Taylor series, the above equation can be proven.

The regret for $E[R_n(T)]$ can be bounded by 
\begin{align*}
\mathbb{E}\left[R_n(T)\right] \leq \sum_{t=1}^{T} \left(3BLT^{- \frac{1}{4}} + At^\theta\right)\\
\leq 3BLT^{\frac{3}{4}} + \frac{A}{1 + \theta}T^{1+\theta}. 
\end{align*}
the detailed proof is attached in appendix.

Combining the above results, the regret \(R(T)\) is bounded by

\vspace{-0.3cm}
\begin{align*}
R(T) \leq (1 - \frac{1}{e})B r^{\text{max}}2 M(M^{\frac{1}{2}}T^{\frac{3}{4}} \log(MT) + T^{\frac{1}{4}}) + \\
\left(1 - \frac{1}{e} \right)B^2r^{\text{max}}M \left( \begin{array}{c} |\mathcal{A}^t_m| \\ B \end{array} \right) \frac{\pi^2}{3} \\
+ 3BLT^{\frac{3}{4}} + \frac{A}{1 + \theta}T^{1+\theta}
\end{align*}

 In order to balance the leading orders, we select the parameters \(z, \gamma,A, \theta\) as following: \(z = \frac{1}{2} , \gamma = \frac{1}{4} , \theta = -\frac{1}{4}\), and \(A = 2Br^{\text{max}} + 2BL\). Thus, the regret \(R(T)\) reduces to

\vspace{-0.3cm} 
\begin{align*}
R(T) \leq &(1 - \frac{1}{e})B r^{\text{max}}2 M(M^{\frac{1}{2}}T^{\frac{3}{4}} \log(MT) + T^{\frac{1}{4}})\\
&+ \left(1 - \frac{1}{e}\right) \cdot MB^2r^{\text{max}} \left(\begin{array}{c} M^{\text{max}} \\ B \end{array} \right) \frac{\pi ^2}{3}\\
&+ \left(3BL + \frac{8}{3}B(r^{\text{max}}+L)\right)T^{\frac{3}{4}}
\end{align*}
 
\end{proof}
{In summary, to derive the upper bound for the regret, we calculate the regret bounds for both exploration and exploitation phase. During the exploration phase, we leverage the bounded nature of submodular function to establish the upper bound. For the exploitation phase, we further divide regret into components arising from suboptimal and near-optimal choices. The Chernoff-Hoeffding bound is applied to determine the upper regret bound associated with suboptimal choices. Lastly, we use Lipschitz-continuous to derive the bound for near-optimal choices. The leading order of the cumulative regret is $O(T^{\frac{3}{4}}\log(T))$, indicating a sublinear growth over the time horizon $T$. This implies that our CCBM scheme exhibits asymptotic optimality and converges toward the optimal strategy.}

\section{Performance Evaluation}
In this section, we evaluate the performance of our CCBM approach through extensive numerical evaluation. We begin by outlining the simulation setup, followed by a comparison of our algorithm with several baseline schemes in terms of multiple performance metrics. 

\subsection{Network settings}

We consider a 3-D indoor network scenario with a size of 40m$\times$40m$\times$3m, consisting of wooden tables, wooden chairs, metal cabinets, and 15 humans randomly moving at a speed of 0.8m/s to emulate dynamic obstacles. In this setup, we place four 60 GHz mmWave APs randomly in the space at a height of 2.9m. Specifically, each AP as the wireless transmitter is equipped with 8 orthogonal beam patterns with equal beam widths, covering a 360$^{\circ}$ azimuth. The environment context $X$ is uniformly divided into 1600 (40$\times$40) small grids, each measuring 1m$^2$. In particular, the neighbouring beams are regarded as arms with the similar context, hence the beams from the same transmitter are categorized into 4 hypercubes, resulting in a total of 16 hypercubes. We employ the commercial ray tracer \textit{Wireless Insite}$^\circledR$\cite{wirelessinsite} to generate realistic network environments and mmWave signal profiles. Additionally, we introduce noise following a normal distribution $\mathcal{N}(0, 5 \rm dB)$ into the obtained received signal strength (RSS) values to account for potential measurement errors in the context information, mirroring the conditions encountered in practical scenarios.

\subsection{Baseline schemes}
We conduct a thorough performance analysis by comparing our approach with the following baseline schemes:
\begin{itemize}
    \item \textbf{Optimal scheme}. This algorithm relies on an oracle search, indicating a priori knowledge of the expected reward $\mu_{a|x}$ for each arm $a$ within $\mathcal{A}^t_m$ at grid $x$. It always selects an optimal subset $\mathcal{S}^*$ to probe the best beam at each time step, offering an upper-bound performance for comparison with other feasible schemes.
    \item \textbf{UCB-based scheme}. This state-of-the-art scheme, proposed in~\cite{10228988}, employs an upper confidence bound approach. In each time step, it strategically selects $B$ arms with the highest estimated upper confidence bounds on their expected rewards. 
    \item \textbf{CC-MAB scheme}: We add the basic contextual MAB algorithm from~\cite{chen2018contextual} as the comparison point. The key distinction with our CCBM approach lies in the fact that CC-MAB incorporates a completely randomized arm selection process during the exploration phase.
\end{itemize}

\subsection{Cumulative regret for beam selection}
\vspace{-0.1cm}
\begin{figure}[htbp]
\centerline{\includegraphics[scale=0.3]{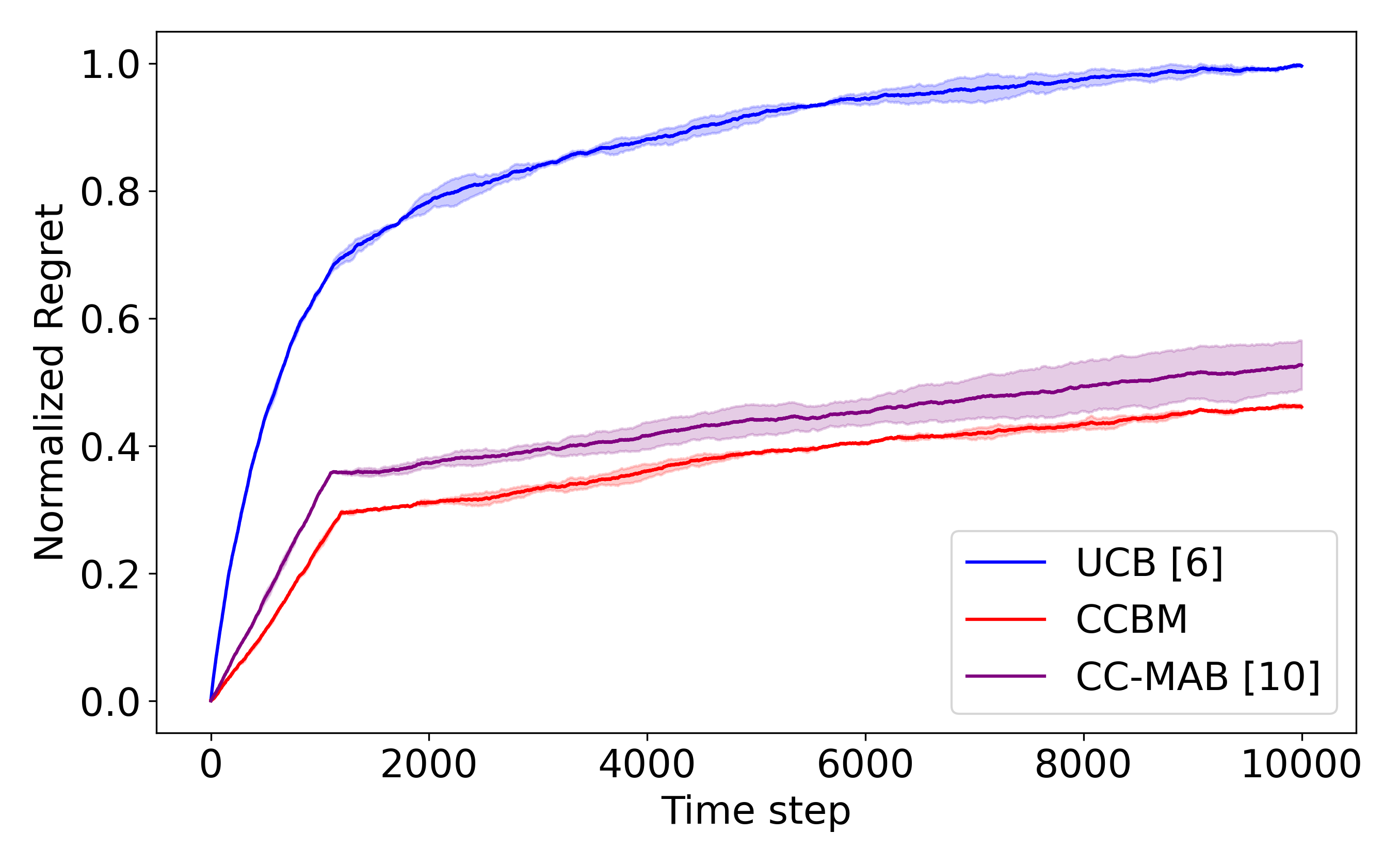}}
\caption{Comparison of regret among different schemes.}
\label{regret}
\end{figure}

To evaluate the disparity between the total reward achieved by a practical probing algorithm and the optimal reward attainable by consistently selecting the best beam, Fig.~\ref{regret} shows the cumulative regret over time for three distinct algorithms. {It is worth noting that cumulative regret is a metric to show the expected cumulative difference between the reward achieved by the optimal algorithm and the designed algorithm.} Obviously, our proposed CCBM exhibits a superior performance compared to the two baselines. Specifically, the UCB-based scheme exhibits the highest regret, consistently maintaining a curve above the others throughout the time horizon. This can be attributed to the fact that it does not account for the properties of a submodular reward function. Further, our CCBM scheme achieves a lower regret than CC-MAB, which demonstrates the effectiveness of incorporating exploitation into the exploration phase via our attention-based selection. Furthermore, a noticeable turning point occurs at time step around 1,600, corresponding to the implementation of our early stopping criterion that prevents extensively useless searches.

\subsection{Beam management overhead and network throughput}

\begin{figure}[htbp]
\centerline{\includegraphics[scale=0.25]{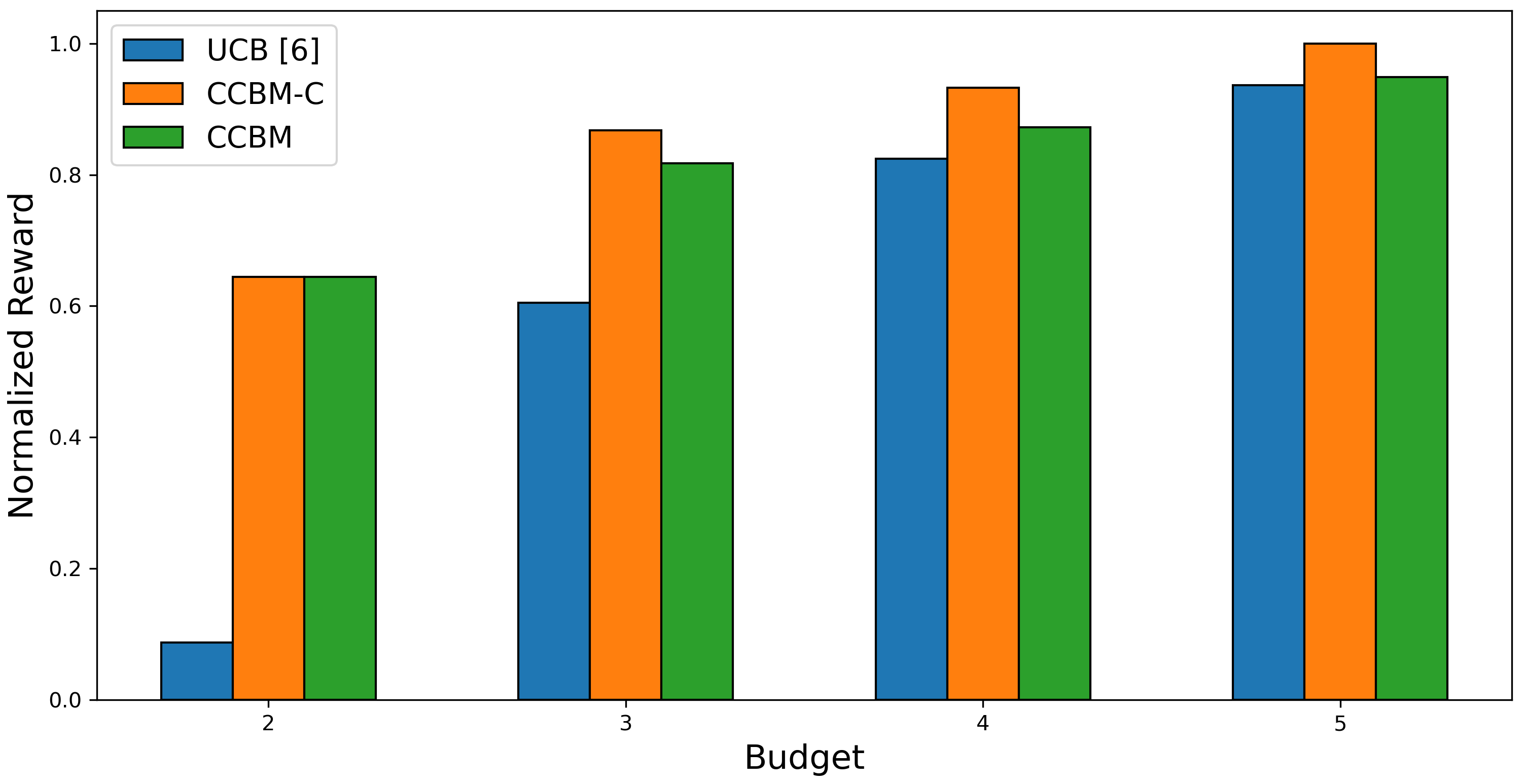}}
\caption{Reward under different beam probing budgets.}
\label{budget}
\end{figure}

In this section, we evaluate the performance in terms of beam management overhead and average user throughput in mmWave networks. First, we consider both CCBM and its variant, CCBM-C, based on the rewards obtained under different probing budgets $B$. The higher $B$ implies a potential manegement overhead. The only difference between the two schemes is that CCBM-C constantly probes beams with a budget of $B$ while CCBM searches a subset of beams with a size of $\frac{B}{2}$ in early stopping phase. As depicted in Fig.~\ref{budget}, an increase in the budget leads to an augmentation in rewards for all the three algorithms. This trend is attributed to the fact that a larger budget enhances the probability of encompassing the optimal beam, thereby increasing the likelihood of identifying the most advantageous beam. Under different budgets, the CCBM-C algorithm consistently secures the highest rewards. Concurrently, CCBM achieves rewards marginally lower than those of CCBM-C while utilizing only \textit{half} of the budget. This demonstrates how efficiently the CCBM algorithm can leverage limited resources to optimize rewards. In this regard, we conclude that both CCBM and CCBM-C can obtain a higher reward at lower overhead, indicating that CCBM is effective in achieving high performance with a constrained beam search budget.

\begin{figure}[htbp]
\centerline{\includegraphics[scale=0.28]{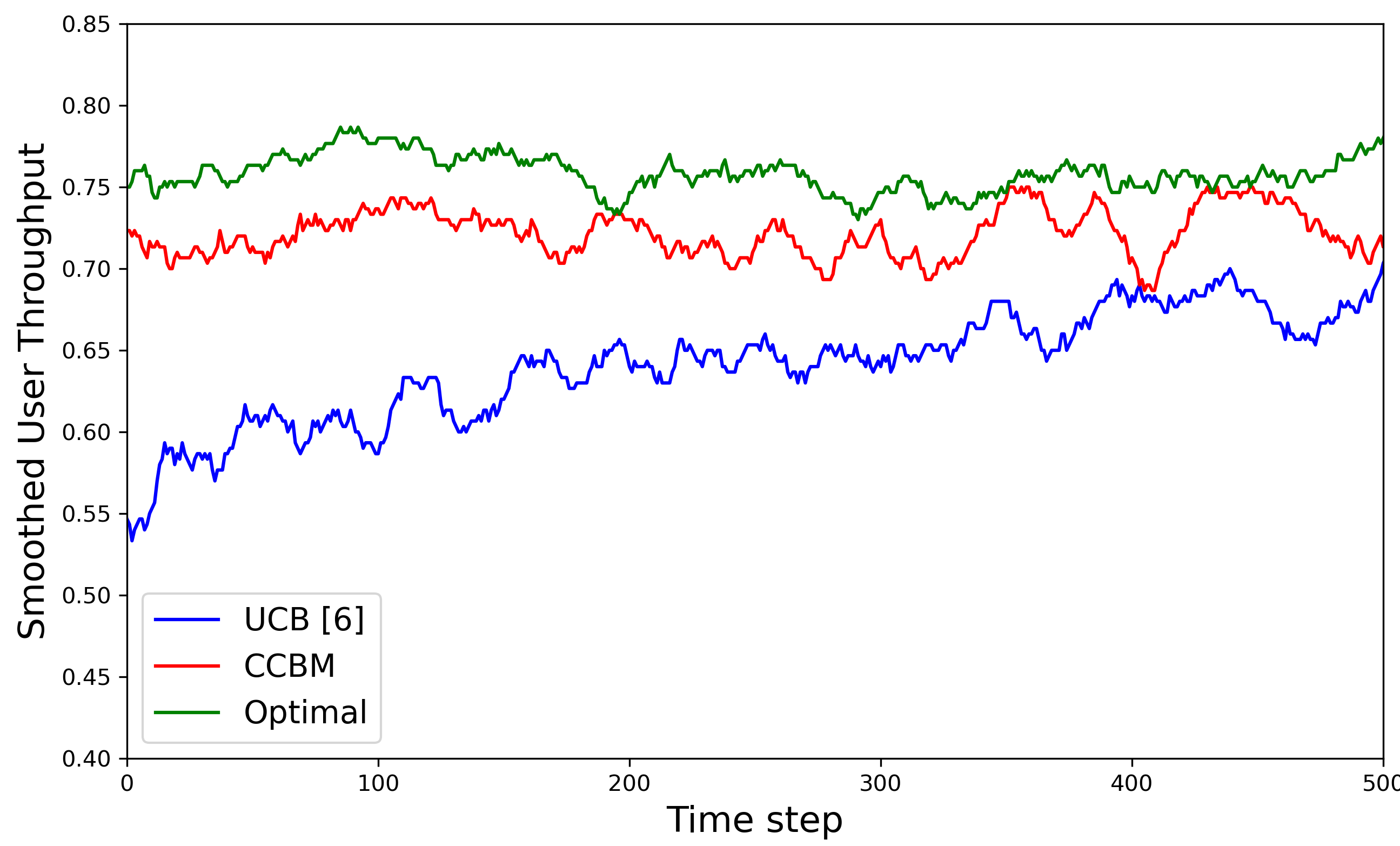}}
\caption{{Comparison of average user throughput among different schemes.} }
\label{throughput}
\end{figure}

Fig.~\ref{throughput} compares the average user throughput of our proposed CCBM against the UCB-based scheme from~\cite{10228988} and the Optimal method with an oracle search. {We apply a sliding window to smooth the evaluation results, making them more readable. The throughput performance is tested over 500 time steps, where at each time step, the user moves to another grid, and the approximate duration for each time step is around 30 ms.} We observe that the throughput of the Optimal scheme is always the highest, benefiting from its priori knowledge about the expected reward of each arm. Consistently, the throughput of CCBM is maintained at a relatively high level, close to the optimal results and significantly surpassing the results of the UCB-based scheme. Additionally, the throughput of CCBM exhibits a lower variance than that of UCB-based scheme, signifying its greater stability. 
Such consistently higher throughput performance underscores the robustness and efficiency of our CCBM scheme in dynamic network environments.
{It is worth noting that localization error is considered throughout the experiments. We simulate this error by introducing white noise that follows a Gaussian distribution. Additionally, the user's location is ultimately mapped into the grids with non-negligible ranges, which inherently provides a degree of error tolerance in our method.}  

Lastly, load balancing is another critical aspect addressed by our CCBM approach. To evaluate the load balancing performance, we utilize the maximum load utilization $L_{max}$ as the metric to qualitatively reflect network congestion~\cite{liu2020blockage}, where a higher $L_{max}$ indicates more server congestion and unbalanced resource usage. The value of $L_{max}$ corresponds to the maximum load among all beams in the network. Fig.~\ref{lmax} shows the average $L_{max}$ across randomly located users. As expected, the Optimal scheme achieves the lowest $L_{max}$ value, while our CCBM approach performs closely to the optimal results. It is observed that the gap is especially smaller under higher density of users in the network, which validates the load balancing capability owing to the strategical reward function.

\vspace{-0.1cm}
\begin{figure}[htbp]
\centerline{\includegraphics[scale=0.4]{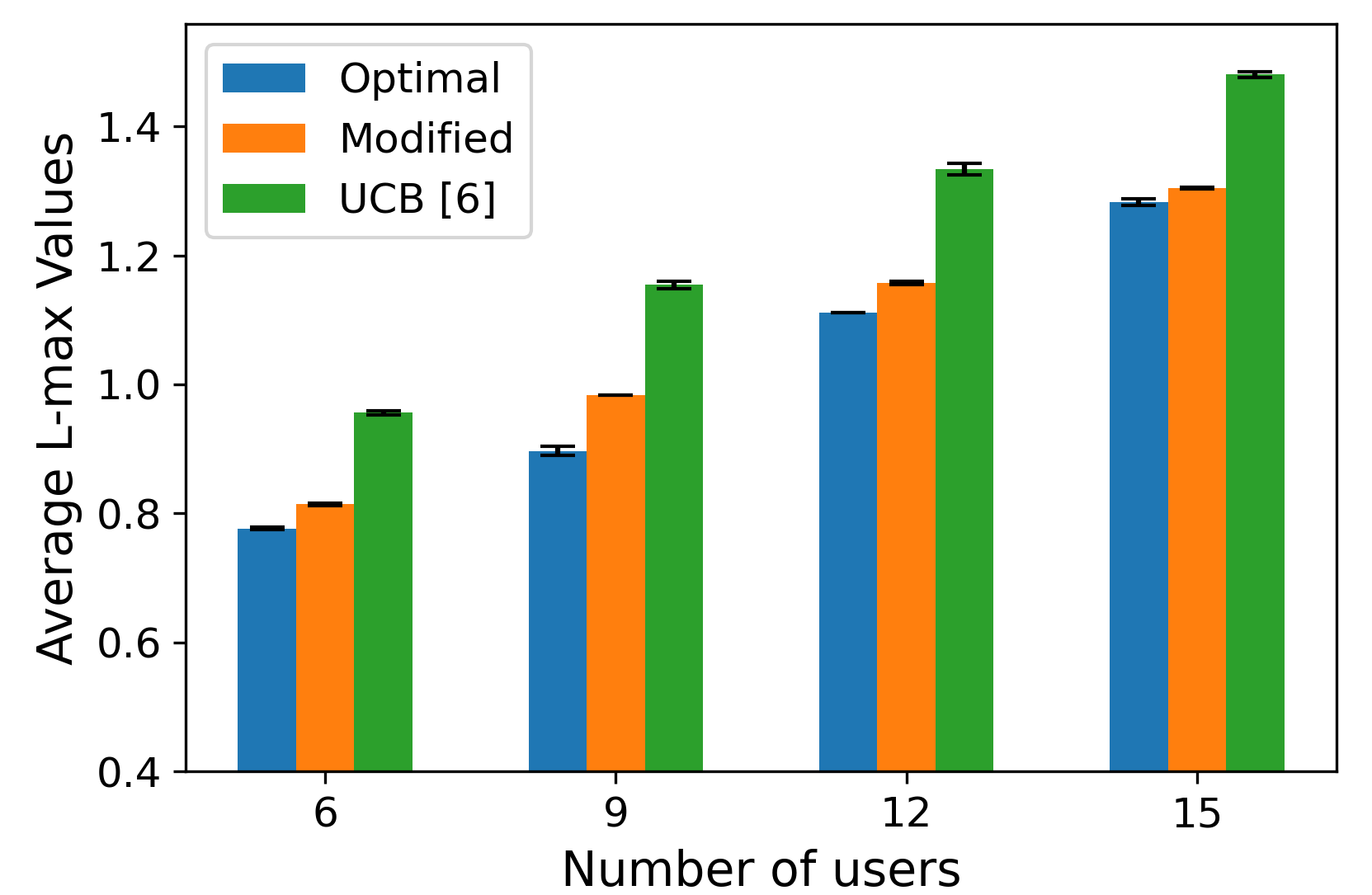}}

\caption{Comparison of $L_{max}$ among different schemes.}
\label{lmax}
\end{figure}
\vspace{-0.2cm}
\subsection{Balancing link-level and network-level performance}

As introduced in Sec. IV, we designate a penalty weight $K$ to denote the maximum number of users that can be connected to a single beam. This weight serves as a penalizing factor, guiding users to pair with an AP with lower load while maintaining relatively high link quality. Consequently, selecting an optimal value for $K$ is crucial to strike a balance between link quality for each pair of transceivers and the overall network load among all transceivers.
In essence, a smaller $K$ value restricts the number of users that can connect to the same beam of an AP, thereby reducing its traffic load. However, this limitation may also prevent a user from accessing an AP with the optimal beam for a superior link quality. Conversely, choosing a larger value for $K$ boosts the overall reward of the MAB algorithm, as one beam can serve more users simultaneously. Nevertheless, this approach may degrade the load balancing performance.

As shown in Fig.~\ref{reward_and_lmax}, we investigate the relationship between beam selection reward and $L_{max}$ across various penalty weight values. It is unsurprising that with increasing $K$ values, the reward rises accordingly as more users can connect to the same beam of an AP, albeit at the expense of the overall load on that AP. 
{Particularly, we observe that the optimal penalty weight is 9 in this evaluation, as evidenced by a sharp increase in reward when $K$ transitions from 2 to 9, accompanied by only a marginal increase in AP load. However, setting $K$ to excessively large values results in a sharp rise in $L_{max}$ as shown in Fig.}~\ref{reward_and_lmax}, {indicating a failure to balance the traffic load over the network. On the other hand, the increase in reward becomes quite marginal when increasing $K$ to a higher value. Therefore, an optimal penalty weight, approximately around 9, successfully strikes a balance between link-level performance and overall network load balancing.}

\begin{figure}[htbp]
\centerline{\includegraphics[scale=0.4]{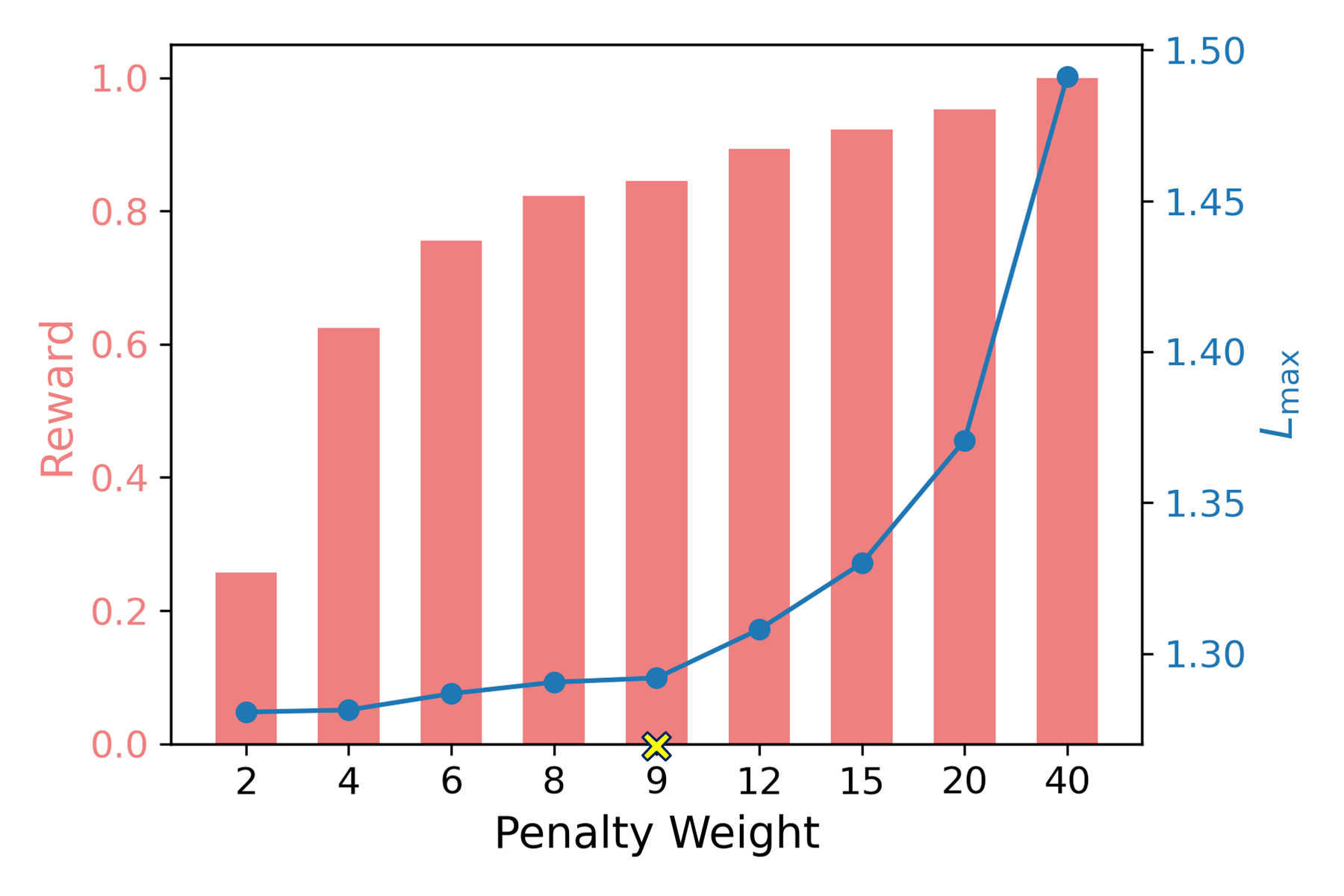}}
\caption{Beam selection reward and network load balancing vs. penalty weights.}
\label{reward_and_lmax}
\end{figure}
 \vspace{-0.5cm}

\section{Conclusions}
This paper presented a contextual combinatorial beam management scheme for joint transceiver paring and beam selection in mmWave wireless networks. Specifically, we explored an early stopping criterion and an attention-based mechanism to mitigate excessive beam search during the online probing phase. Theoretical analysis established its asymptotic optimality by setting an upper bound on cumulative regret of the MAB algorithm. Additionally, a carefully designed reward function accounted for load balancing among deployed APs within the network, facilitating the selection of a globally optimal AP and beam combination, thereby enhancing overall network performance across multiple dimensions. Through a comprehensive series of evaluations and theoretical analyses, our CCBM has demonstrated superiority over other baseline schemes in optimizing beam management in dense mmWave wireless networks, thus paving the way for the development of next-generation multiple access and advanced transceivers. 

\section*{Acknowledgment}
This research was supported by the National Science Foundation through Award CNS--2312138, CNS--2312139, CNS-2415208, CCSS-2434053 and ECCS 2434054.

\vspace{-0.0cm}
{\appendix
Here we derive a bound for \( E[R_n(T)] \).
For each time slot \( 1 \leq t \leq T \), the regret of selecting near-optimal subsets can be expressed as

\vspace{-0.3cm}
\begin{align*}
R_n(T) = \sum_{t=1}^{T} I_{\{W^t, S^t \in S_b \setminus \mathcal{L}^t(p^t)\}} \times \\
 \left( \left( 1 - \frac{1}{e} \right) \cdot R( S^{*, t}(x^t),r^t) - R(S^t,r^t) \right).
\end{align*}

\noindent Define \( Q_t = W^t \cap \{S^t \in S_b \backslash \mathcal{L}^t(p^t)\} \) to signify the event that a near-optimal arm set is selected. Then, we have

\vspace{-0.3cm}
\begin{align*}
E[R_n(T)] = \sum_{t=1}^{T} E[I_{\{Q_t\}} \times \\
\left( (1 - \frac{1}{e}) \cdot R(S^{*,t}(x^t),r^t) - R(S^t,r^t) \right)]. 
\end{align*}

By applying the principle of conditional expectation, this is equivalent to
\vspace{-0.3cm}
\begin{align*}
    &E[R_n(T)]= \sum_{t=1}^{T} \text{Prob}\{Q(t)\} \cdot \\
    &E\left[ \left(1 - \frac{1}{e}\right) \cdot R(S^{*,t}(x^t),r^t) - R(S^t,r^t) \middle| Q(t)\right]\\
    &\leq \sum_{t=1}^{T} E \left[ \left(1 - \frac{1}{e}\right) \cdot R( S^{*,t}(x^t),r^t) - R( S^t,r^t) \middle| Q(t)\right].
\end{align*}

When \( Q(t) \) holds true, which means that the algorithm enters an exploitation phase and define \( J \in S_b \backslash \mathcal{L}^t(p^t) \). By the definition of \( \mathcal{P}^{ue,t} \), it holds that \( C^t(p^t_{m}) > K(n_x) = n_x^{\frac{1}{2}} \log(n_x) \) for all \( p^t_m \in p^t \). Moreover, given that \( J \in S_b \backslash \mathcal{L}^t(p^t) \), it holds
\[
R(\tilde{S}^t(p^t),\underline{\mu}_p^t) - R(J,\bar{\mu}_p^t) < An_x^\theta 
\]
To establish an upper bound on the regret, we consider this expression
\[
\sum_{t=1}^{T} \mathbb{E} \left[ \left(1 - \frac{1}{e}\right) \cdot R( S^{*,t}(x^t),r^t) - R(J,r^t) \middle| Q(t) \right]
\]
\[
= \sum_{t=1}^{T} \left((1 - \frac{1}{e}) \cdot R( S^{*,t}(x^t),\mu^t) - R( J,\mu^t)\right). 
\]
Through application of the Lipschitz-continuous condition, we have:
\begin{align*}
&\left(1 - \frac{1}{e}\right) \cdot R(S^{*,t}(x^t),\mu^t) - R( J,\mu^t)\\
&\leq \left(1 - \frac{1}{e}\right) \cdot R(S^{*,t}(x^t),\tilde{\mu}_{p}^t) + BLh_T^{-1} - R(J,\mu_{x}^t)\\
&\leq \left(1 - \frac{1}{e}\right) \cdot R(\tilde{\mu}_{p}^t, S^{*,t}(p^t)) + BLh_T^{-1} - R(J,\mu_{x}^t)\\
&\leq R(\tilde{S^t(p^t),\mu}_{p}^t) + BLh_T^{-1} - R( J,\mu_{x}^t)\\
&\leq R(\tilde{S^t(p^t),\mu}_{p}^t) + 2BLh_T^{-1} - R( J,\mu_{x}^t)\\
&\leq R(\tilde{S^t(p^t),\mu}_{p}^t) + 3BLh_T^{-1} - R( J,\mu_{x}^t)\\
&\leq 3BLh_T^{-1} + An_x^\theta \leq 3BLh_T^{-1} + At^\theta,
\end{align*}
where the third inequality follows the definition of \( \tilde{S}^{*,t}(p^t) \) and \( \tilde{S}^t(p^t) \). Using \( h^{-1}_T = \lceil T^{\frac{1}{4}} \rceil^{-1} \leq T^{-\frac{1}{4}} \), we further have

\[
\mathbb{E} \left[R(S^{*,t}(x^t),r^t) - R(J,r^t) | Q(t)\right] \leq 3BLT^{-\frac{1}{4}} + At^\theta. 
\]

\noindent Therefore, the regret can be bounded by
\begin{align*}
\mathbb{E}\left[R_n(T)\right] \leq \sum_{t=1}^{T} \left(3BLT^{- \frac{1}{4}} + At^\theta\right)\\
\leq 3BLT^{\frac{3}{4}} + \frac{A}{1 + \theta}T^{1+\theta}. 
\end{align*}

}

\bibliographystyle{IEEEtran}
\bibliography{reference}

\vfill

\end{document}